\documentclass[12pt]{article}
\usepackage[margin = 1in]{geometry}
\usepackage{amsfonts}
\usepackage{amsmath}
\usepackage{amssymb}
\usepackage[english]{babel}
\usepackage{amsthm}
\usepackage{tikz}
\usepackage{float}
\usepackage[export]{adjustbox}
\usepackage{subcaption}
\usepackage{bbm}
\usepackage{comment}
\usepackage{natbib}
\usepackage{enumerate}
\usepackage{hyperref}
\usepackage{cleveref}
\usepackage{arydshln}
\usepackage{pgfplots}
\usepackage{pgfplots}
\usepackage{setspace}
\usepackage{graphicx}
\usepackage[export]{adjustbox}
\pgfplotsset{compat=1.18}
\usetikzlibrary{calc}
\usetikzlibrary{patterns}
\usetikzlibrary{arrows}
\setcitestyle{authoryear,open={(},close={)}}

\newcommand{\E}{ \mathbbm{E}}
\newcommand{\R}{ \mathbbm{R}}

\excludecomment{note}

\theoremstyle{definition}
\newtheorem{theorem}{Theorem}
\newtheorem{proposition}{Proposition}
\newtheorem{lemma}{Lemma}

\newtheorem{corollary}{Corollary}
\newtheorem*{definition}{Definition}

\theoremstyle{remark}
\newtheorem{remark}{Remark}

\DeclareMathOperator*{\argmin}{\arg\min}
\DeclareMathOperator*{\argmax}{\arg\max}

\DeclareMathOperator{\cvh}{\text{cvh}}

\newcommand{\und}{\underline}
\newcommand{\vst}{\vspace{3mm}}

\usepackage{inputenc}

\title{Obviously Strategy-Proof Multi-Dimensional Allocation and the Value of Choice}
\author{Quitz\'{e} Valenzuela-Stookey\thanks{Department of Economics, University of California, Berkeley. Email: quitze@berkeley.edu.}}
\date{January 27, 2026}

\begin{document}
\maketitle

\begin{abstract}
   A principal must allocate a set of heterogeneous tasks (or objects) among multiple agents. The principal has preferences over the allocation. Each agent has preferences over which tasks they are assigned, which are their private information. The principal is constrained by the fact that each agent has the right to demand some status-quo task assignment. I characterize the conditions under which the principal can gain by delegating some control over the assignment to the agents. Within a large class of delegation mechanisms, I then characterize those that are obviously strategy-proof (OSP), and provide guidance for choosing among OSP mechanisms. 
\end{abstract}

Consider a principal who needs to allocate a set of tasks among a set of agents, but is constrained by the fact that agents' preferences over tasks are unobservable. For example, a hospital seeks to implement a new mechanism for allocating cases to doctors \citep{handel2025thinking}. The hospital has an estimate of the performance of each doctor on different types of cases, and would like to use this information to reform the existing case-assignment system. However, doctors also have preferences over the cases that they handle, and will resist the adoption of a new mechanism if it leads to a less-preferred assignment. Knowing the doctors' preferences would enable the hospital to assign cases more efficiently while ensuring that doctors are not made worse off. However, as doctors' preferences are their private information the hospital faces a trade-off: in order for the allocation of cases to respond to doctors' preferences, the hospital must relinquish some control over the assignment, and thus respond less directly to performance metrics. Is this trade-off worth making, and if so, how should the system be designed?

Such trade-offs arise in a number of settings, including the assignment of Children's Protective Services investigators to cases \citep{baron2025mechanism}, tax collectors to households \citep{bergeron2024supermodular}, and workers to tasks \citep{acemoglu2018modeling, valenzuela2025automation}. Moreover, the word ``tasks'' need not be interpreted literally; similar forces arise in other allocation problems, such as the matching of refugees to locations \cite{bansak2018improving, ahani2021placement, delacretaz2023matching}. Formally, these are multi-dimensional screening problems without transfers, in which the designer has a known objective but is constrained by the need to ensure participation in the face of agents' unknown preferences.\footnote{The primary focus here is on problems in which monetary transfers between the designer and the agents are either infeasible (as in school choice or organ exchange, \citealt{roth2007repugnance}) or beyond the scope of the particular mechanism-design question. For example, a firm may seek to improve the allocation of tasks to workers whose wages have already been determined by a collective bargaining agreement. That said, it is not difficult to extend the model to include transfers, as discussed in \Cref{sec:transfers}.} In this paper I study two general conceptual questions that arise in the context of such task-allocation problems. First, when is there value in using mechanisms that respond to agents' preferences over tasks? In other words, when is it optimal to grant agents the ability to make choices that influence the allocation? Second, if there is value in choice-based mechanisms, which ones should we use in practice? 

The large literature on multi-dimensional screening has shown that identifying optimal mechanisms is general analytically intractable (e.g. \citealt{thanassoulis2004haggling, manelli2006bundling}), and even characterizing incentive constraints can be challenging (e.g. \cite{rochet1987necessary, lahr2024extreme}). Nonetheless, it turns out to be possible to characterize the conditions under which choice-based mechanisms are optimal. Moreover, I make progress on the normative question of which mechanisms to use by characterizing the set of obviously strategy-proof (OSP) mechanisms within a restricted class of \textit{trading mechanisms}. 

The setting is as follows. There are $I$ different types of tasks, and $J$ agents (both finite). Completing a task of type $i$ with agent $j$ generates a social cost which is known to the designer, and potentially heterogeneous across tasks and agents. In addition, agents have preferences over the tasks they are assigned, which are their private information. The goal of the designer is to allocate tasks among the agents so as to minimize the total social cost. However, the designer is constrained by agents' preferences: each agent has an outside option assignment of tasks, which they have the right to demand.\footnote{For example, an agent's outside option could represent the terms of their current contract. More abstractly, the outside option could reflect the designer's notion of a fair workload: the designer wishes to assign tasks efficiently, without unfairly burdening any individual agent. See for example \cite{baron2025mechanism}.}

Trading mechanisms are a flexible class of mechanisms which partially delegate the choice of the allocation to the agents. Such a mechanism specifies, for each agent, $j$, a menu of \textit{trading sets}, $\{ T_j^k\}_{k \in K_j}$, where each trading set $T_j^k$ is a subset of $\R^I_+$, and $K_j$ is an arbitrary index set (which may be finite or infinite). A trading set is a set of possible bundles of tasks that $j$ can receive. The designer asks each agent to choose a trading set from their menu, and commits to a selection rule that respects these choices: given the profile of chosen trading sets, the designer must select an allocation such that each agent's assigned bundle of tasks lies in their chosen set. On one extreme, the designer can offer agent $j$ the singleton menu consisting of all assignments that $j$ would prefer to their outside option under any possible preference. Doing so for all $j$ is equivalent to choosing an allocation without eliciting information from the agents. By offering larger menus the designer gains flexibility in responding to agent's preferences in exchange for sacrificing some control over the final allocation. Intuitively, choice relaxes the participation constraint, but introduces new incentive constraints. The class of trading mechanisms is rich; it includes the ``binary classification mechanisms'' studied by \cite{baron2025mechanism}, as well as variants of widely-studied market-based mechanisms (e.g. \cite{hylland1979efficient}, \cite{budish2011combinatorial}).

A mechanism is obviously strategy-proof if (loosely) the worst outcome an agent could obtain if they follow the strategy prescribed by the mechanism is better than the best possible outcome from any deviation \citep{li2017obviously}. OSP is often interpreted as a notion of simplicity, as it implies that the solution to the decision problem faced by participants is easy to identify.\footnote{In lab experiments \cite{li2017obviously} shows that subjects are more likely to play dominant strategies in obviously strategy-proof mechanisms. See also \cite{breitmoser2022obviousness}.} Simplicity has particular normative appeal in settings such as the current one, in which preferences and allocations are potentially high-dimensional objects. 

The first main result, \Cref{thm:choice}, provides necessary and sufficient conditions for the optimality of (non-trivial) choice-based mechanisms among OSP trading mechanisms. In fact, \Cref{thm:choice} shows that in terms of determining whether choice has value, the restriction to OSP trading mechanisms is generally without loss: if choice is optimal when we can use any mechanism (not just trading mechanisms) satisfying the weaker Bayesian incentive compatibility (BIC) constraint, then it is also optimal among OSP trading mechanisms. 

The conditions under which choice has value are mild, and should be satisfied in most applications. The question then is which mechanisms to use. The second main result, \Cref{thm:osp_trade}, characterizes the set of OSP trading mechanisms. Effectively, all and only OSP trading mechanisms take the following form. Either an agent is given no choice (i.e. they are offered a singleton menu consisting of all allocations that are unambiguously preferred to their outside option) or they are offered a menu in which each trading set is a ray with an endpoint at the status-quo allocation. In addition, the set of rays that constitute an agent's menu must satisfy a ``polarization'' condition, which constrains the degree of similarity between any two rays. We refer to these as \textit{polarized ray} mechanisms.  

Any polarized ray mechanism, in addition to being obviously strategy-proof, satisfies a number of desirable properties. The size of the menu that each agent is offered is at most $I + 1$ (\Cref{cor:polarizing_dimension}). Moreover, all that the agents need to know about the mechanism is the form that their trading sets take; additional details about the designer's selection rule or the trading sets of other agents are irrelevant from their perspective. Thus the designer is free to choose any selection rule, and in doing so optimally can  guarantee a (weak) improvement over the status-quo allocation (i.e. the allocation in which all agents receive their outside-option assignment). This is true regardless of the distribution over agents' preferences, or even the form that these preferences take. Thus the mechanism is robust to model misspecification in these dimensions. 

Given these properties, any polarized ray mechanism would be a reasonable candidate to use in practice. On the other hand, solving for the optimal OSP mechanism is in general a difficult computational problem. While the ability to restrict attention to polarized-ray mechanisms simplifies the problem significantly, the optimization program remains challenging, and I do not solve it here. However, I provide some guidance for choosing among polarized ray mechanisms by studying a ``large market'' relaxation of the optimal design problem. In this relaxation, the dual program has a convenient formulation, which can facilitate the search for an optimal mechanism. 

\vst
\noindent\textit{Related literature}

\vst
The task-allocation problem studied here is a multi-dimensional screening problem (where the word ``task'' should be understood broadly).\footnote{In particular, a screening problem with a type-dependent participation constraint and no transfers.} \cite{budish2012matching} surveys the approaches to such problems, and categorizes them broadly as either ``good properties'' or ``mechanism design''. Under the good-properties (or axiomatic) approach, the designer specifies a set of desirable properties (e.g. stability, Pareto efficiency) and identifies a mechanism that satisfies them. Deferred acceptance \citep{gale1962college}, top trading cycles \citep{shapley1974cores}, and market-based mechanisms \citep{varian1973equity, hylland1979efficient, budish2011combinatorial} are perhaps best understood in this way. In contrast, the mechanism-design approach is to ``maximize an objective subject to constraints''. In a single-dimensional setting, the optimal-auctions approach of \cite{myerson1981optimal} is a classical example. In multi-dimensional settings, however, solutions of the mechanism-design form are limited; extensive work has demonstrated the technical difficulty of these problems \citep{pavlov2011optimal, daskalakis2014complexity, daskalakis2015strong, hart2019selling}. Even characterizing incentive compatibility in a tractable way has proved challenging \citep{rochet1987necessary, lahr2024extreme}.

The current paper makes (modest) progress on the mechanism-design approach in a particular multi-dimensional screening setting.\footnote{There is also an empirical literature studying the value of choice-based mechanisms in other settings. For example, \cite{agarwal2025choices} study this question in the context of deceased-donor kidney waitlists.} The approach here builds on two basic premises. First, I focus on obvious strategy-proofness, as opposed to weaker solution concepts such as Bayesian incentive compatibility (the standard in the mechanism-design approach). Obvious strategy-proofness has been widely studied due to its normative appeal, and has been identified as a promising avenue in multi-dimensional settings \citep{palacios2024combinatorial}.\footnote{\cite{palacios2024combinatorial} write ``We argue that the behavioral implications for future developments moving toward OSP [combinatorial allocation] formats should not be underestimated.''} Studies of OSP in other multi-dimensional problems include \cite{feldman2014combinatorial, dutting2016posted, golowich2021computational}. However, work to analytically characterize OSP in such settings has been limited, and OSP has not been applied to the type of task allocation problem studied here. 

Second, I restrict attention primarily to trading mechanisms.\footnote{An exception to both the OSP requirement and the restriction to trading mechanisms is \Cref{thm:choice}, which applies to all Bayesian incentive compatible mechanisms.} While this restriction sacrifices some generality, it allows me to fully characterize obvious strategy-proofness. Moreover, trading mechanisms are still a rich class of mechanisms, and possess a number of normatively-appealing properties. Trading mechanisms can be understood as a form of delegation \citep{holmstrom1977incentives, holmstrom1984theory}. Typically, delegation mechanisms are applied in settings where the agent possesses private information that is directly relevant for the principal's decision (e.g. \cite{alonso2008optimal}). Here the agents' private information is only indirectly relevant (in that the designer needs to ensure participation) and so the value of delegation is less obvious. Nonetheless, I show that non-trivial delegation can be valuable, and can in fact achieve gains whenever it is possible to do so with more sophisticated choice-based mechanisms. 

Finally, this paper is related to other work on task allocation problems \citep{baron2025mechanism, valenzuela2025automation}. The setting is most similar to that of \cite{baron2025mechanism}. There, the authors study the design of Bayesian incentive-compatible mechanisms within a restricted class of binary-classification mechanisms, in which there are effectively only two types of tasks. It turns out, however, that the solution their program is (essentially) a particular polarized-ray mechanism, in which each agent's menu has only three elements. The current paper shows that within a much broader class, only polarized-ray mechanisms are admissible if the goal is obvious strategy-proofness. \cite{baron2025mechanism} also introduce a two-step procedure to solving a large-market version of their program, which we generalize in \Cref{sec:lm}.

The remainder of the paper is organized as follows. I present the model in \cref{sec:model}. I then identify the conditions under which choice is valuable (\cref{sec:valueofchoice}) and characterize OSP trading mechanisms (\cref{sec:osp}). In \cref{sec:lm} I show how studying a large-market version of the model simplifies the optimal-design problem. \Cref{sec:conclusion} concludes. Proofs omitted from the main text are presented in the appendix.

\section{Model}\label{sec:model}

There is a finite set of tasks (or objects), indexed by $i \in \mathcal{I} = {1,\dots,I}$ (with $I \geq 2$) and a finite set of agents, indexed by $j \in \mathcal{J} = \{1,\dots, J\}$. An allocation is a matrix $Z \in \mathbbm{Z}_{\geq 0}^{\mathcal{I}\times\mathcal{J}}$ such that row $i$ sums to $n^i$, where $n^i$ is the number of instances of task $i$ that need to be completed.\footnote{Where $ \mathbbm{Z}_{\geq 0}$ denotes the non-negative integers, inclusive of zero.} Let $\mathcal{Z}$ be the set of allocations. A random allocation is a distribution over allocations. For a random allocation $\mu \in \Delta(\mathcal{Z})$, let $\mu_j \in \R^I_+$ denote the expected assignment for agent $j$, where $\mu_j(i)$ is the expected number of task $i$ assigned to $j$. For convenience, where it will not cause confusion I refer to elements of $\Delta(\mathcal{Z})$ simply as allocations, and to $\mu_j$ as $j$'s \textit{assignment}.

The designer's objective is to minimize the expected total social cost, given by 
\[
\sum_{i \in \mathcal{I}}\sum_{j \in \mathcal{J}} \pi_j(i) \mu_j(i),
\]
for some known weights $\pi_j(i)$.\footnote{Describing the objective as minimizing social cost, as opposed to maximizing welfare, is simply a normalization.} We can think of $\pi_j(i)$ as representing the (observable) performance of agent $j$ on task $i$.

On the other hand, the designer does not know the preferences of agents over tasks. Let $c_j(i) > 0$ be the private cost to agent $j$ of completing task $i$. Then the total \textit{workload} for $j$ from random allocation $\mu$ is
\[
\sum_{ i \in \mathcal{I} } c_j(i) \mu_j(i)
\]
and we maintain the normalization that agents prefer lower workload.\footnote{The substantive assumption is $c_j(i) > 0$, i.e. all tasks have the same qualitative effect on agents' payoffs.} Let $\mathbf{C}$ be the set of preference profiles, with typical element $\mathbf{c} = (c_j(\cdot))_{j=1}^J$. Let $F_j$ be the distribution of $j$'s preference type $c_j$. We maintain the assumption that types are independent across agents, and that $F_j$ has full support on a hypercube $C_j \subset \R^I_+$. 

The agents' outside option is defined by a random allocation $\sigma \in \Delta(\mathcal{Z})$, which we refer to as the status quo. Each agent has the right to demand their marginal allocation from the status quo, $\sigma_j$. This participation decision occurs at the interim stage, i.e. after each agent observes their own preferences but before they engage with the designer's mechanism.\footnote{In particular, if the designer's mechanism involves randomization, the participation decision takes place before this randomization realizes. It fact, OSP trading mechanisms, characterized below, satisfy a stronger ex-post participation constraint which holds after all types are revealed (but before the designer's randomization resolves).}

If the designer could implement any random allocation in $\Delta(\mathcal{Z})$ then the social-cost minimization problem would be a straightforward linear program. However, the designer here is constrained by the fact that $i)$ agents can demand their status-quo assignment and, $ii)$ the designer does not know the agents' preferences. The question is whether the designer can improve upon $\sigma$, in terms of the expected social cost, by using a choice-based mechanism (and if so, how). Before answering these questions, we introduce the class of trading mechanisms. 

\subsection{Trading mechanisms}

A \textit{trading protocol} consists of an \textit{endowment} $\eta \in \Delta(\mathcal{Z})$; a menu of \textit{trading sets} $\{T_j^k\}_{k \in K_j}$ for each agent $j$, where $T_j^k \subset \R^I_+$ and $K_j$ is an arbitrary index set (which may be finite or infinite); and a \textit{selection rule} $P: \prod_{j \in \mathcal{J}} K_j \rightarrow \Delta(\mathcal{Z})$ such that 
\[
P_j((k_j)_{j\in \mathcal{J}}) \in T_j^{k_j} \quad \forall \quad (k_j)_{j\in \mathcal{J}} \in \prod_{j \in \mathcal{J}} K_j, \ j \in \mathcal{J},
\]
where $P_j$ denotes the assignment for $j$ under $P$. The protocol works as follows: each agent chooses a trading set $k \in K_j$, and given the profile of choices $k_{\mathcal{J}} = (k_j)_{j\in \mathcal{J}}$ the allocation is determined by the selection rule $P(k_{\mathcal{J}})$. The essential restriction is that the selection rule must respect the agents' choices of trading sets. By definition a trading protocol must be feasible, in that such a selection must exist. Additionally, each agent has the right to demand their endowment, so $\{\eta_j\} \in \{T_j^k\}_{k \in K_j}$ for all $j$. Note that the status quo, $\sigma$, does not appear in this definition. While we are primarily interested in trading protocols where $\eta = \sigma$, we define trading protocols more generally here in order to separate the incentive and participation constraints.

Given a trading protocol $\mathcal{T}$, a strategy for agent $j$ is a map $s_j: C_j \mapsto K_j$. We refer to a pair $(\mathcal{T}, s_{\mathcal{J}})$, where $\mathcal{T}$ is a trading protocol and $s_{\mathcal{J}} = \{s_j: C_j \mapsto K_j\}_{j \in \mathcal{J}}$ is a strategy profile, as a \textit{trading mechanism}.

\vst
\noindent\textbf{Examples.} We describe three examples of trading mechanisms. 

\begin{itemize}
    \item \textit{Bilateral trading mechanisms.} The simplest trading mechanism, defined by $\gamma \in \R^I$ and two agents $j',j''$ as follows: agent $j'$ is offered the option keep their status quo or trade to $\sigma_{j'} + \gamma$, while agent $j''$ is offered the option to trade from $\sigma_{j''}$ to $\sigma_{j''} - \gamma$. Trade occurs if and only if both $j'$ and $j''$ wish to do so. All other agents receive their status quo. The natural strategies are for agent $j'$ (agent $j''$) to trade if and only if $\gamma \cdot c_{j'} \leq 0$ ( $\gamma \cdot c_{j''} \geq 0$).

     To generalize the bilateral trading mechanisms, one margin is to move beyond bilateral trade by allowing for exchanges among multiple agents. Another margin is to allow for a larger range of trades, by offering more and/or larger trading sets.

    \item \textit{Two-price binary-classification mechanisms.} In a binary-classification mechanism, introduced by \cite{baron2025mechanism}, there are (effectively) two tasks, $A$ and $B$.\footnote{Formally, given an arbitrary set of tasks a binary partition of tasks into classes $A$ and $B$ is defined, and the mechanism only takes values $\mu \in \Delta(\mathcal{Z})$ such that $\mu$ and $\sigma$ have the same marginal distributions conditional on each of the two classes. In other words, each agent's allocation can be described by pair $(n^A_j,n^B_j) \in \R^2_+$, where the agent receives $n^k_j$ tasks drawn from class $k \in \{A,B\}$ according to the conditional distribution of $\sigma$.} A two-price binary-classification mechanism is one in which each agent can is given two prices $(p_j^1,p_j^2)$ and can report whether they would like to trade for (trade away) task $A$ at a rate of $p_j^1$ fewer ($p_j^2$ more) of task $B$. This defines a trading protocol where the menu for $j$ is $\{ \{\sigma_j \}, \{ (\sigma_j(A), \sigma_j(B)) 
    + x(1, -p_j^1) : x \in \R_+ \}, \{ (\sigma_j(A), \sigma_j(B)) + x(-1,p_j^2) : x \in \R_+ \} \}$.
    
    \item \textit{Pseudo-market mechanisms.} The natural benchmark for this type of combinatorial allocation problem is the class of market-based mechanisms \citep{varian1973equity,hylland1979efficient, budish2011combinatorial}. The idea would be to endow each agent with their status quo $\sigma_j$, and specify a ``price'' for each task in units of a fictional currency. Agents would then be free to buy and sell tasks at the specified prices. In the standard approach, prices are allowed to adjust in response to realized demands in order to clear the market. As individual agents have non-zero price impact, such mechanisms are only approximately strategy-proof in large (finite) markets. Such pseudo-market mechanisms are not trading protocols, but we can specify a trading-protocol version of the mechanism in which a price vector $p$ \textit{fixed} up front, agents submit their demanded task assignments given $p$, and the designer then executes only the trades that can be accommodated under market clearing. In this trading protocol, for every point $x$ on the budget $\{x \in \R^I_+ : p \cdot x = p \cdot \sigma_j\}$, there is $k \in K_j$ such that $\{x, \sigma_j\} = T_j^k$. We could generalize this trading protocol by allowing agents to submit non-degenerate demand sets, where trading sets would take the form $\{A,\sigma_j\}$ for $A \subset \{x \in \R^I_+ : p \cdot x = p \cdot \sigma_j\}$.
\end{itemize} 

The goal is to characterize the set of obviously strategy-proof trading mechanisms. 

\begin{definition}
    Given a trading protocol $\mathcal{T}$, an action $k \in K_j$ is obviously dominant for agent $j$ with type $c_j$ if 
    \[
    \sup_{k_{-j} \in K_{-j}} \ c_j \cdot P_j(k, k_{-j}) \leq \inf_{k_{-j} \in K_{-j}} c_j \cdot P_j(k', k_{-j})
    \]
    for all $k' \in K_j$.
\end{definition}
\begin{definition}
    Given a trading protocol $\mathcal{T}$, a strategy profile $\{s_j: C_j \mapsto K_j\}_{j \in \mathcal{J}}$ is \textit{obviously strategy-proof} (OSP) if $s_j(c_j)$ is obviously dominant for all $j \in \mathcal{J}$ and $c_j \in C_j$. We say that a trading mechanism $(\mathcal{T}, s_{\mathcal{J}})$ is OSP if $s_j$ is OSP given $\mathcal{T}$ for all $j$.
\end{definition}

Observe that an OSP trading mechanism with $\eta = \sigma$ satisfies the agents' participation constraint ex-post, for any realized type profile, rather than merely at the interim stage.\footnote{If the mechanism involves randomization by the designer, the ex-post stage here occurs before this randomization is realized.}

\subsection{Discussion of modeling assumptions}\label{sec:transfers}

\subsubsection*{Why OSP?}

Having a mechanism that is OSP is useful for practical implementation, as it reduces the cognitive load for agents \citep{li2017obviously}. Moreover, while tractable characterizations of Bayesian incentive-compatibility in multi-dimensional screening problems have remained elusive, it turns out that we can give a simple characterization of OSP trading mechanisms.

Even if the designer would be willing to use a non-OSP mechanism, OSP mechanisms are useful for computing a lower bound on the value of choice: the fact that agents have an obviously dominant strategy means that the best-response problem for each agent is relatively easy to solve, whereas to estimate the gains from a mechanism that is not strategy proof (obviously or not) requires solving a fixed-point problem to identify an equilibrium of the induced game.

\subsubsection*{Why trading mechanisms?}

While the class of trading mechanisms is large and fairly flexible,  I do not claim that it contains all interesting and relevant mechanisms. Again, one reason to restrict attention to trading mechanisms is tractability. As discussed above, progress on the ``mechanism-design'' approach to multi-dimensional allocation problems has been limited by their intractability, but within the class of trading mechanisms we are able to make some progress. Beyond tractability, trading mechanisms have appealing properties for practical applications. For one, they are relatively simple to describe. Moreover, OSP trading mechanisms turn out to be robust to model misspecification regarding the distribution of agents' preferences, about which the designer may in practice have only limited information. 

\subsubsection*{Transfers}

As discussed above, the focus here is on settings in which transfers between the designer and agents are either infeasible \citep{roth2007repugnance} or beyond the scope of the mechanism-design exercise. While there are many applications which fit this description, there are certainly others in which transfers are an important part of the designer's toolkit. 

The model can be extended to incorporate transfers by letting one of the ``tasks'', $m$, be payments made by the agents to the designer. This (potentially) requires two changes to the model. First, we have assumed that all tasks have the same qualitative effect on agents' payoffs. If tasks are burdensome then we can continue to use the normalization $c_j(m) > 0$. However, if tasks are ``goods'' then $m$ would have the opposite effect, so we would need $c_j(m) < 0$. (We have also assumed that allocations are non-negative; we may wish to relax this restriction on $m$ in some settings so that the designer can make payments to agents). Second, we have assumed a fixed supply of each task. In a setting with transfers, this would correspond to requiring that the mechanism generate a fixed level of revenue. In many applications, we may wish to relax this constraint, which changes the space of feasible mechanisms. 

In summary, under some conditions we can incorporate transfers directly into the current model (e.g. tasks are burdensome, payments must be non-negative, and revenue must equal some predetermined level). More generally, the modifications needed to incorporate transfers do not appear to be large. However, we leave the analysis of these cases for future work. 

\section{When is choice useful?}\label{sec:valueofchoice}

We first answer the question of whether sacrificing some control of the allocation to the agents can help the designer achieve a lower expected social cost. We are interested in the answer to this question both when the designer is restricted to OSP trading mechanisms, and when they can use any Bayesian incentive compatible (BIC) mechanism. It is instructive, moreover, to also consider the value of choice under a relaxation of the status-quo constraint. Let $m$ be a mechanism, where $m(\mathbf{c}) \in \Delta(\mathcal{Z})$ is the distribution over allocations when agents report type profile $\mathbf{c}$, and write $m_j(i|\mathbf{c})$ for the expected number of $i$ that goes to agent $j$. Under the relaxed status-quo constraint, the designer is constrained to mechanisms such that  
\[
 \sum_{i \in \mathcal{I}} c_j(i) \E_{-j}[m_j(i|c_j,c_{-j})] \leq \beta_j \sum_{i \in \mathcal{I}} c_j(i) \sigma_j(i) \ \ \forall \ j \in \mathcal{J}, \ c_j \in C_j,
\]
where $\beta_j \geq 1$. In other words, no agent can be given an expected workload that is more than $\beta_j$ times their expected workload under the status quo. The baseline case is $\beta_j = 1$, while $\beta_j > 1$ allows for slack in this constraint. 

Consider the general BIC mechanism design problem, where aside from the status quo constraint and the market-clearing condition (i.e. every task instance is assigned to exactly one agent), we impose no further restrictions on the mechanism used to allocate tasks. By the standard revelation principle, it is without loss of optimality to restrict attention to incentive-compatible direct-revelation mechanisms. Then the designer's program is 
\begin{align}
    \min_{m} & \ \E \left[ \sum_{i=1}^I\sum_{j=1}^J \pi_j(i) \ m_j(i|\mathbf{c}) \right] \label{prog:1} \\
    s.t. \quad &\sum_{i \in \mathcal{I}} c_j(i) \E_{-j}[m_j(i|c_j,c_{-j})] \leq \beta_j \sum_{i \in \mathcal{I}} c_j(i) \sigma_j(i) \ \ \forall \ j \in \mathcal{J}, \ c_j \in C_j \tag{SQ}\\
    & \sum_{ i \in \mathcal{I} } c_j(i) \  \E_{-j} [ m_j(i|c_j,c_{-j})] \leq \sum_{ i \in \mathcal{I} } c_j(i) \ \E_{-j}[ m_j(i|c_j', c_{-j})] \quad \forall \ j \in \mathcal{J}, \ c_j, c_{j}' \in C_j \tag{IC}\\
    & \sum_{j \in \mathcal{J}} m_j(i|\mathbf{c}) =  n^i \quad \forall \ i \in \mathcal{I}, \ \mathbf{c} \in \mathbf{C} \tag{MC}
\end{align}

The designer here faces a multi-dimensional screening problem with a type-dependent participation constraint. While the solution to this problem is an open question, it turns out that we can answer the more limited question of whether or not the optimal mechanism must respond non-trivially to agents' preferences.  

Say that \textit{choice improves upon the status quo} if there is a mechanism $m$ that satisfies the IC, SQ, and MC constraints, is not constant in the preference profile, and yields a strictly lower expected social cost than the status quo. Say that \textit{choice is optimal} if all solutions to designer's program in \cref{prog:1} are non-constant mechanisms.\footnote{From a design perspective, we could first solve for the optimal deterministic mechanism and set this as the status quo, in which case the only question is whether choice can improve upon the status quo. However, it is also useful to understand when choice is optimal under a given status quo constraint.}

Our first result shows that in many circumstances we should expect choice to be optimal. Say that \textit{the status quo has full support} if $\sigma_j(i) > 0$ for all $i \in \mathcal{I}, j\in \mathcal{J}$.  Say that \textit{preferences have a common support} if $C_j = C_{j'}$ for all $j,j' \in \mathcal{J}$ (the distributions may still differ across agents). When preferences have a common support, it must be non-degenerate for choice to play any role. Say that preferences have \textit{bounded support} if $C_j$ is a bounded hypercube. Let $\Bar{c}^i_j = \max\{c_j(i) : c_j \in C_j \}$ and $\und{c}^i_j = \min\{c_j(i) : c_j \in C_j \}$, so $[\und{c}_j^i,\Bar{c}_j^i]$ is the domain of $j$'s preference in dimension $i$. Say that \textit{performance is identical} for all agents if $\pi_{j'} = \pi_{j''}$ for all $j',j'' \in \mathcal{J}$.

\begin{theorem}\label{thm:choice}
    If preferences have a common bounded support and the status quo has full support, then the following are equivalent 
    \begin{enumerate}[i.] 
        \item Agents' preference distributions are non-degenerate, and performance is not identical across agents.
        \item Choice improves upon the status quo among OSP trading mechanisms.
        \item Choice improves upon the status quo among all BIC mechanisms, for any $(\beta_j)_{j \in \mathcal{J}}$.
    \end{enumerate}
     If, moreover, $\beta_j = 1$ for all $j$, and the support of agents' preferences has dimension $I$, then the following is also equivalent to $i.$ - $iii.$
     \begin{enumerate}[i.]
     \setcounter{enumi}{3}
         \item Choice is optimal among OSP trading mechanisms (and a fortiori among all BIC mechanisms).\footnote{``Choice is optimal among OSP trading mechanisms'' means that every solution to the designer's program when they are additionally restricted to use OSP trading mechanisms is a non-constant mechanism. }
     \end{enumerate}
\end{theorem}

Condition \textit{i.} in \Cref{thm:choice} is clearly a necessary condition for choice to improve upon the status quo; if preference distributions are degenerate there is no role for choice, and if performance is identical there is no scope for improvement. The equivalence between \textit{i.} \textit{ii.}, and \textit{iii.} shows that under fairly mild conditions (common support of preferences and full-support status quo) choice improves upon the status quo whenever the problem is non-trivial, whether or not the designer is restricted to OSP trading mechanisms. Moreover, in the baseline case where there is no slack in the SQ constraints and agents preferences are not degenerate along any dimension, choice is also optimal whenever the problem is non-trivial. 

\subsection{Proof of \texorpdfstring{\Cref{thm:choice}}{}}

To understand \Cref{thm:choice}, it is useful to separate it into constituent parts. First, consider the gap between ``choice improves upon the status quo'' and ``choice is optimal''.  We relate this gap to the degrees of slack in the SQ constraints, parameterized by $\beta_j$, and the range of preferences, $[\und{c}_j^i, \Bar{c}_j^i]$. Define $\hat{c}^i := \max_{j \in \mathcal{J}} \und{c}_j^i$, so $\hat{c}^i \in [\und{c}_j^i, \Bar{c}_j^i]$ for all $i,j$ if and only if $\cap_{j \in \mathcal{J}} C_j \neq \varnothing$.

\begin{lemma}\label{lem:gap_improve_optimal}
    Assume $\cap_{j \in \mathcal{J}} C_j \neq \varnothing$. If an allocation $m \in \Delta(\mathcal{Z})$ (i.e. a non-choice-based mechanism) satisfies (SQ) and (MC) then 
    \[
    \sum_{i \in \mathcal{I}} \sum_{j \in \mathcal{J}}\left( \Bar{c}_j^i - \hat{c}^i \right) \left(m_j(i) - \beta_j\sigma_j(i) \right)^+ \leq \sum_{i \in \mathcal{I}} \sum_{j \in \mathcal{J}} \hat{c}^i \left(\beta_j - 1\right) \sigma_j(i).
    \]
\end{lemma}
\begin{proof}
    Assume that $m \in \mathcal{Z}$ satisfies (SQ), i.e.
    \begin{equation}      
    \sum_{i \in \mathcal{I}} m_j(i) c_j(i) \leq \beta_j\sum_{i \in \mathcal{I}} \sigma_j(i) c_j(i) \quad \forall \ j \in \mathcal{J}, \ \forall \ c_j \in C_j \tag{SQ}. 
    \end{equation}
    Then for any $\mathbf{c} \in \mathbf{C}$ we must have 
    \begin{align}
        0 &\geq \sum_{i \in \mathcal{I}} \sum_{j\in \mathcal{J}} \left(m_j(i) - \beta_j\sigma_j(i) \right) c_j(i) \notag \\
        & = \sum_{i \in \mathcal{I}} \hat{c}^i \sum_{j\in \mathcal{J}} \left(m_j(i) - \beta_j\sigma_j(i) \right) + \sum_{i \in \mathcal{I}} \sum_{j\in \mathcal{J}} \left(m_j(i) - \beta_j\sigma_j(i) \right) \left(c_j(i) - \hat{c}^i\right) \label{eq:sq_opt2}
    \end{align}
    where the inequality in the first line follows by summing the (SQ) constraint over $j$, and the second line is simple algebra. If $m$ satisfies (MC) then the first term in \cref{eq:sq_opt2} is equal to 
    \[
    \sum_{i \in \mathcal{I}} \sum_{j \in \mathcal{J}} \hat{c}^i \left(1 - \beta_j\right) \sigma_j(i)
    \]
    Then in for \cref{eq:sq_opt2} to hold, it must be that 
    \[
        \sum_{i \in \mathcal{I}}\sum_{j \in \mathcal{J}} \hat{c}^i \left(\beta_j - 1\right) \sigma_j(i) \geq \sum_{i \in \mathcal{I}} \sum_{j\in \mathcal{J}} \left(m_j(i) - \beta_j\sigma_j(i) \right) \left(c_j(i) - \hat{c}^i\right)
    \]
    for all $\mathbf{c} \in \mathbf{C}$. Under the assumption that $\cap_{j \in \mathcal{J}} C_j \neq \varnothing$ we have $\Bar{c}_j^i \geq \hat{c}^i$ for all $i \in \mathcal{I}, j \in \mathcal{J}$. Then setting $c_j(i) = \hat{c}^i$ if $\left(m_j(i) - \beta_j\sigma_j(i) \right) < 0 $ and $c_j(i) = \Bar{c}^i_j$ otherwise, we have
    \[
        \sum_{i \in \mathcal{I}}\sum_{j \in \mathcal{J}} \hat{c}^i \left(\beta_j - 1\right) \sigma_j(i) \geq \sum_{i \in \mathcal{I}} \sum_{j\in \mathcal{J}} \left(\Bar{c}_j^i - \hat{c}^i\right) \left(m_j(i) - \beta_j\sigma_j(i) \right)^+
    \]
as desired. 
\end{proof}
\begin{corollary}\label{cor:improve_optimal}
    If $\beta_j= 1$ for all $j \in \mathcal{J}$ and $\cap_{j \in \mathcal{J}} C_j$ has dimension $I$ then the only non-choice-based mechanism satisfying (SQ) and (MC) is the status quo. In particular, choice is optimal if and only if choice improves upon the status quo. 
\end{corollary}
\begin{proof}
    From \Cref{lem:gap_improve_optimal} for $\beta_j =1$, (SQ) and (MC) imply that for any non-choice-based mechanism $m \in \Delta(\mathcal{Z})$, 
    \[
    \sum_{i \in \mathcal{I}} \sum_{j \in \mathcal{J}}\left( \Bar{c}_j^i - \hat{c}^i \right) \left(m_j(i) - \sigma_j(i) \right)^+ \leq 0.
    \]
    If $\cap_{j \in \mathcal{J}} C_j$ has dimension $I$ then $\Bar{c}^i_j > \hat{c}^i$ for all $i,j$. If $m$ satisfies (MC) and $m \neq \sigma$ then $m_j(i) - \sigma_j(i) > 0$ for some $i,j$, so the above inequality is violated. Thus the only deterministic mechanism satisfying (SQ) and (MC) is the status quo. 
\end{proof}

Now consider the equivalence between \textit{i.}, \textit{ii.}, and \textit{iii.} in \Cref{thm:choice}. Underlying this equivalence is the simple class of bilateral trading mechanisms, introduced above, which can be used to improve upon the status quo. Let $j',j''$ be the agents involved in the mechanism, let $\gamma \in \R^I$ be the trade, and assume that agent $j'$ (agent $j''$) follows the strategy of trading if and only if $\gamma \cdot c_{j'} \leq 0$ ( $\gamma \cdot c_{j''} \geq 0$). We begin with the following simple observation.

\begin{lemma}\label{lem:bilateral}
    The bilateral trading mechanism for agents $j',j''$ with trade $\gamma$ is feasible and improves upon the status quo if and only if
    \begin{enumerate}[i.]
        \item $\sigma_{j'} + \gamma \geq 0$ and  $\sigma_{j''} - \gamma \geq 0$,
        \item $(\pi_{j'} - \pi_{j''})\cdot \gamma < 0$, 
        \item $\min \left\{ \gamma \cdot c_{j'} : c_{j'} \in C_{j'} \right\} < 0$,
        \item $\max \left\{ \gamma \cdot c_{j''} : c_{j''} \in C_{j''} \right\} > 0$.
    \end{enumerate}
\end{lemma}

\begin{proof}
    Trade is feasible by condition $i.$ Agent $j'$ (resp. $j''$) is willing to trade if and only if $\gamma \cdot c_j' \leq 0$ (resp. $\gamma \cdot c_{j''} \geq 0$). Thus conditions $iii.$ and $iv.$ ensure that trade occurs with positive probability (given the maintained assumption that the distribution has full support). Condition $ii.$ implies that when trade does occur, it reduces the social cost. 
\end{proof}

The real question is whether we can find a bilateral trading which satisfies these conditions. The following sufficient conditions guarantee that such a mechanism exists, and immediately implies the conclusions of \Cref{thm:choice}.

\begin{proposition}\label{prop:bilateral_exist}
    If there exist $j',j''$ such that
    \begin{enumerate}[i.]
        \item $\pi_{j'} \neq \pi_{j''}$; 
        \item $C_{j'} \cap C_{j''}$ is non-empty and non-degenerate; and
        \item $\sigma_{j'}(i) > 0$ and $\sigma_{j''}(i) > 0$ for all $i \in \mathcal{I}$;
    \end{enumerate}
    then there is a bilateral trading mechanism that improves upon the status quo.
\end{proposition}
\begin{proof}
    We begin with an intermediate result. 
    
    \begin{lemma}\label{lem:choice1}
        If $\pi_j' \neq \pi_{j''}$; $\sigma_{j'}(i), \sigma_{j''}(i) > 0$ for all $i \in \mathcal{I}$; and there exist points $c',c'' \in \cap_{j\in\mathcal{J}} C_j$ such that $\cvh(c'',0) \cap \cvh(\pi_{j'} - \pi_{j''} \ , \ c' \ , \ 0 ) = \{ 0\}$; then there is a bilateral trading mechanism satisfying the conditions of \Cref{lem:bilateral}. 
    \end{lemma}
    \begin{proof}
        First, because $\sigma_{j'}(i),\sigma_{j''}(i) > 0$ for all $i$, if we can find $\gamma$ satisfying $ii. - iv.$ then we can scale it to satisfy $i.$

        Assume without loss of generality that there exists $i \in \mathcal{I}$ such that $\pi_{j'}^i - \pi_{j''}^i > 0$ (otherwise simply reverse the roles of $j'$ and $j''$). Then $0$ is an extreme point of $\cvh(\pi_{j'} - \pi_{j''} , c' , 0 )$. Thus $\cvh(\pi_{j'} - \pi_{j''} , c' , 0 ) \setminus \{0\}$ and $\cvh(c'',0)\setminus \{ 0\}$ are disjoint, non-empty, and convex. By the separating hyperplane theorem there exists $\gamma \in \R^I$ such that $(\pi_{j'} - \pi_{j''}) \cdot \gamma < 0$, $\gamma \cdot c' < 0$, and $\gamma \cdot c'' > 0$. Because $\min \left\{ \gamma \cdot c_{j'} \in C_{j'} \right\} \leq \gamma \cdot c'$ and $\max \left\{ \gamma \cdot c_{j''} \in C_{j''} \right\} \geq \gamma \cdot c''$, conditions $ii.-iv.$ of \Cref{lem:bilateral} are satisfied. 
    \end{proof}

    To prove \Cref{prop:bilateral_exist}, it remains to show that if $C_{j'} \cap C_{j''}$ is non-empty and non-degenerate then there exist points $c',c'' \in C_{j'} \cap C_{j''}$ satisfying the conditions of \Cref{lem:choice1}.
    
    Let $\delta = \pi_{j'} - \pi_{j''}$. Because $C_{j'} \cap C_{j''}$ is a non-degenerate hypercube, there are $c^1, c^2 \in C_{j'} \cap C_{j''}$ such that $c^1,c^2,\delta$ are pairwise linearly independent. It suffices to show that either $\cvh(c^1,0) \cap \cvh(\delta, c^2, 0) = \{0\}$ and/or $\cvh(c^2,0) \cap \cvh(\delta, c^1, 0) = \{0\}$.
    
    To see this, note that $\cvh(c^1,0) \cap \cvh(\delta, c^2, 0) \neq \{0\}$ if and only if $c^1$ is in the convex cone generated by $(\delta, c^2)$. That is, there exist $\alpha,\rho \geq 0$ such that $c^1 = \alpha c^2 + \rho \delta$. Moreover, $\alpha,\rho > 0$ because $c^1,c^2,\delta$ are pairwise linearly independent. Symmetrically $\cvh(c^2,0) \cap \cvh(\delta, c^1, 0) \neq \{0\}$ if and only if there exist $\hat\alpha,\hat\rho > 0$ such that $c^2 = \hat\alpha c^1 + \hat\rho \delta$. Suppose, towards a contradiction, that both are true. Then
    \begin{align*}
    c^1 &= \alpha(\hat{\alpha} c^1 + \hat{\rho} \delta ) + \rho\delta \\
    &= \alpha\hat{\alpha} c^1 + (\alpha\hat{\rho} + \rho) \delta 
    \end{align*}
    which contradicts the linear independence of $c^1$ and $\delta$.
\end{proof}

\Cref{prop:bilateral_exist} shows that bilateral trading mechanisms work to improve upon the status quo under mild conditions. The most restrictive condition of \Cref{prop:bilateral_exist} is that $\sigma_{j'}(i) > 0$ and $\sigma_{j''}(i) > 0$ for all $i \in \mathcal{I}$. However even this can be relaxed, at the cost of additional notation, by restricting attention only to the dimensions such that $\sigma_{j'}(i) > 0$ and $\sigma_{j''}(i) > 0$.  

We have established that there is, in general, value in using mechanisms which respond to agents' preferences. In particular, we have shown that gains can be obtained within a simple class of bilateral trading mechanisms. The next step towards quantifying the value of choice is to understand more generally which mechanisms can be used.   

\section{Obviously strategy-proof allocation}\label{sec:osp}

Throughout this section, assume $C_j = \R^I_+$ unless otherwise stated. In particular, we say that $j$'s preferences are \textit{strictly monotone} if $C_j = \R^I_{++}$.

We use $(\mathcal{T},s_{\mathcal{J}})$ and $(\hat{\mathcal{T}}, \hat{s}_{\mathcal{J}})$ to denote distinct trading mechanisms, where $s_{\mathcal{J}} = (s_j)_{j\in \mathcal{J}}$, $\hat{s}_{\mathcal{J}} = (\hat{s}_j)_{j\in \mathcal{J}}$, $\mathcal{T} = (\eta_j, ((T_j^k)_{k \in K_j})_{j \in \mathcal{J}}, P)$ and $\hat{\mathcal{T}} = (\hat{\eta}_j, ((\hat{T}_j^k)_{k \in \hat{K}_j})_{j \in \mathcal{J}}, \hat{P})$.

\begin{definition}
    A trading mechanism $(\mathcal{T}, s_{\mathcal{J}})$ is \textit{dominated} by trading mechanism $(\hat{\mathcal{T}}, \hat{s}_{\mathcal{J}})$ if for all $j \in \mathcal{J}$
    \begin{enumerate}
        \item $\eta_j = \hat{\eta}_j$;
        \item $K_j = \hat{K}_j$ and $s_j = \hat{s}_j$; and 
        \item $T_j^k \subset \hat{T}_j^{k}$ for all $k \in K_j$; and
    \end{enumerate}
\end{definition}
In other words, in the dominating mechanism the strategies and endowments are unchanged, but each trading set is expanded. Note that we place no restrictions on the designer's selection rule. In the dominating mechanism the designer has more flexibility in choosing the allocation. 

A trading protocol has redundancies if some trading sets $T_j^k$ contain allocations that are never chosen under $P$. Such redundancy is irrelevant from the points of view of both the designer and the agents. Given a trading protocol $\mathcal{T}$, we define it's \textit{essential reduction} to be the protocol with the same endowment and selection rule, but where each trading set $T_j^k$ is replaced by $\{ x \in T_j^k : x = P_j(k,k_{-j}) \text{ for some } k_{-j} \in K_{-j}\}$. We refer to the mechanism $(\tilde{\mathcal{T}},s_{\mathcal{J}})$ as the essential reduction of $(\mathcal{T},s_{\mathcal{J}})$ when $\tilde{\mathcal{T}}$ is the essential reduction of $\mathcal{T}$. We say that a mechanism is \textit{non-redundant} if it is equal to its essential reduction. We say that mechanism $(\mathcal{T},s_{\mathcal{J}})$ is \textit{effectively dominated} by $(\hat{\mathcal{T}},s_{\mathcal{J}})$ if the essential reduction of $(\mathcal{T},s_{\mathcal{J}})$ is dominated by $(\hat{\mathcal{T}},s_{\mathcal{J}})$.

Define the \textit{myopically-optimal choices} for agent $j$ as
\[
a_j^*(c_j|\mathcal{T}) := \argmin_{k \in K_j} \left\{ \min \left\{ c_j \cdot x \ : \ x \in \tilde{T}_j^k \right\} \right\}. 
\]
where $\tilde{T}_j^k$ is the essential reduction of $T_j^k$. In other words, if agent $j$ with preference $c_j$ could choose their allocation freely from $\cup_{k \in K_j}T_j^k$ (ignoring redundant allocations), it would be optimal to choose an $x \in T_j^{k}$ for $k \in a_j^*(c_j)$. We say that an agent's strategy is myopically optimal if they always make such choices. Then we have the following simple observation.

\begin{lemma}\label{lem:myopic}
    If a trading mechanism is OSP then agents' strategies are myopically optimal.
\end{lemma}
\begin{proof}
    Suppose not, then there exists $x \in \tilde{T}_j^k$ for $k \neq s_j(c_j)$ such that $c \cdot x < c \cdot x'$ for all $x' \in T_j^{s_j(c_j)}$. But then $s_j(c_j)$ does not obviously dominate $k$. 
\end{proof}

\Cref{lem:myopic} narrows the range of OSP mechanisms to those with myopically-optimal strategies, but tells us nothing about the form that the trading protocol should take. The following class of trading mechanisms turn out to play an essential role in this analysis. 

\begin{definition}
    The trading sets for agent $j$ are \textit{polarized rays} if
    \begin{enumerate}
        \item For every $k\in K_j$, the set $T_j^k$ is a non-monotone ray with endpoint $\eta_j$, i.e. $T_j^k = \{ x \in \R^I : (x - \eta_j) = y \kappa_j^k, \ y \in \R_+ \}$ for some $\kappa_j^k \in \R^I$ such that $\kappa_j^k \not \geq 0$ and $\kappa_j^k \not \leq 0$, and 
        \item For any $k',k'' \in K_j$, there exist scalars $y',y'' \geq 0$, with at least one strictly positive, such that $y'  \kappa_j^{k'} + y''  \kappa_j^{k''}  \geq 0$
    \end{enumerate}
\end{definition}

Intuitively, polarization of the rays defining an agent's trading sets means that these sets move the agent's allocation in different directions. One way to see this is that a necessary (but not sufficient) condition for polarization is that no two rays can ``take away'' the same task, i.e. the sets of negative coordinates for any two rays must be disjoint. \Cref{fig:polarized_rays} depicts an example. 

\begin{figure}[h]
    \centering
    \includegraphics[width=0.5\linewidth]{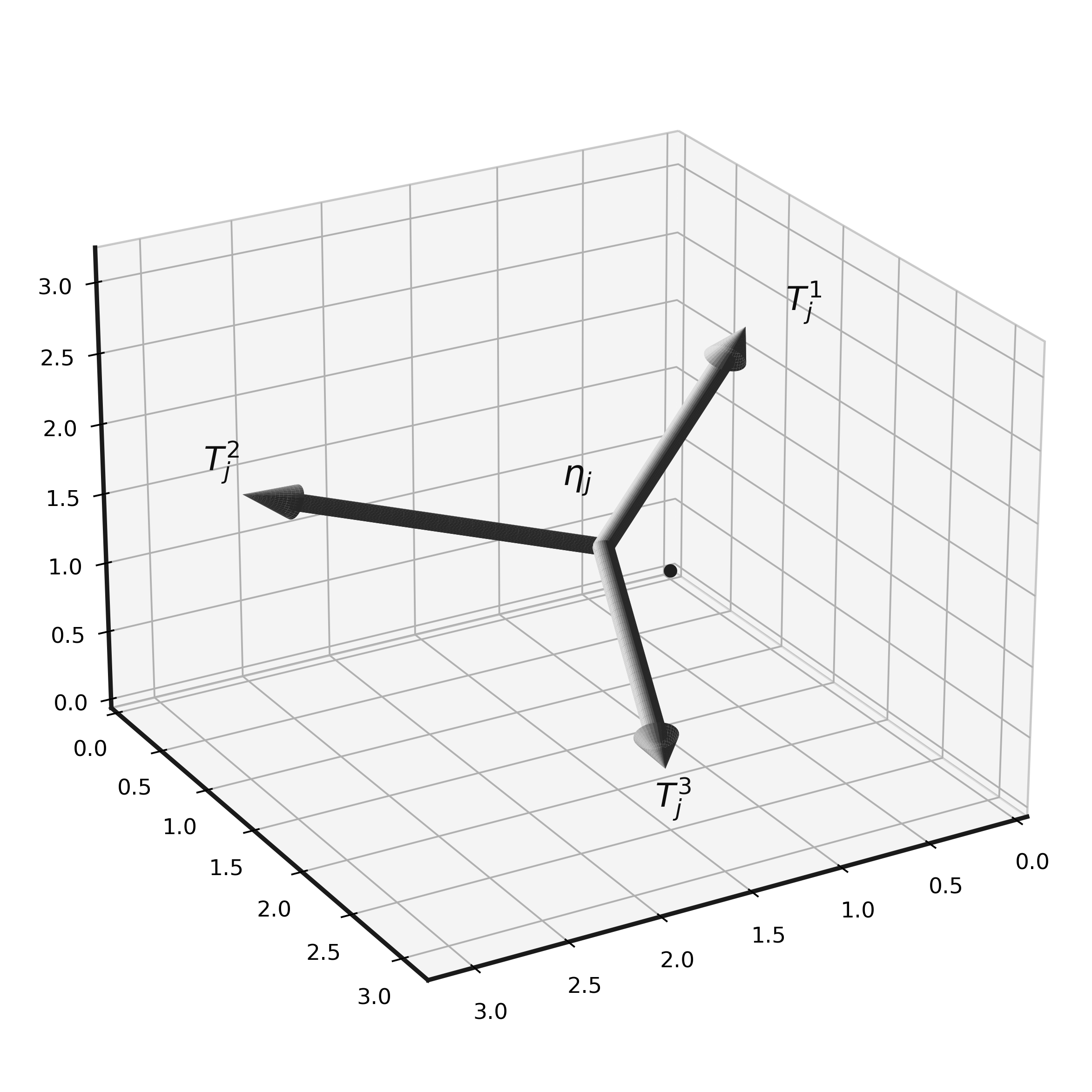}
    \caption{Polarized rays}
    \label{fig:polarized_rays}
\end{figure}

\begin{definition}
    A \textit{polarized ray mechanism} consists of a trading protocol $\mathcal{T}$ and a strategy profile $\{s_j: C_j \mapsto K_j\}_{j \in \mathcal{J}}$ such that for any agent $j$
    \begin{enumerate}
        \item either $j$'s trading sets are polarized rays, or $j$ receives two trading sets: $\{\eta_j\}$ and $\{x \in \R_+ \ : \ x \leq \eta_j\}$; and
        \item $j$'s strategy is myopically optimal, i.e. $s_j(c_j) \subset a^*_j(c_j)$ for all $c_j \in C_j$.
    \end{enumerate}
\end{definition}

\noindent\textbf{Examples.} From the examples introduced above
\begin{itemize}
    \item A bilateral trading mechanism is (trivially) a polarized ray mechanism. 
    \item A two-price binary-classification mechanism is a polarized ray mechanism if and only if $p_j^1 \leq p_j^2$ for all $j$.
\end{itemize}

\begin{theorem}\label{thm:osp_trade}
    Any polarized ray mechanism is OSP. If preferences are strictly monotone then any OSP trading mechanism is effectively dominated by a polarized ray mechanism. 
\end{theorem}

\begin{proof}
Assume that we have a mechanism $(\mathcal{T}, a^*_{\mathcal{J}})$ that is OSP.  We begin with some intermediate claims. 

    \noindent\textit{Claim 1.} $\eta_j \in T_j^k$ for all $j \in \mathcal{J}$ and $k \in K_j$.
    
    \noindent\textit{Proof of Claim 1.} If every agent $j' \neq j$ chooses $\{\eta_{j'}\}$ then market clearing requires $j$ be assigned $\eta_j$. Thus it must be that $\eta_j \in T_j^k$ for all $k \in K_j$.  \hfill\qedsymbol

    \noindent\textit{Claim 2.} For any $c_j$ it must be that $c_j\cdot x \geq c_j \cdot \eta_j$ for all $k \neq s_j^*(c_j)$ and $x \in T_j^k$.

    \noindent\textit{Proof of Claim 2.} Suppose there exists $k \neq s_j^*(c_j)$ and $x \in T_j^k$ such that $c_j \cdot x < c_j \cdot \eta_j$. By Claim 1, $\eta_j \in T_j^{s_j(c_j)}$. But then $\eta_j$ could be chosen following action $s_j(c_j)$ while $x$ could be chosen for action $k$, so $s_j(c_j)$ does not obviously dominate $k$. \hfill\qedsymbol

    \noindent\textit{Claim 3.} If $x',x'' \in T^k$ then there cannot exist $c_j \in \R^I_+$ such that $c_j \cdot x' < c_j \cdot \eta_j$ and $c_j \cdot x'' > c_j \cdot \eta_j$. 

    \noindent\textit{Proof of Claim 3.} By Claim 2, if such a $c_j$ exists then $s_j(c_j) = k$, because $c_j \cdot x' < c_j \cdot \eta_j$. But because $x'' \in T_j^k$ and $c_j \cdot x'' > c_j \cdot \eta_j$, choosing $k$ does not obviously dominate choosing $\{\eta_j\}$.  \hfill\qedsymbol

    \noindent\textit{Claim 4.} For each $j \in \mathcal{J}$ and $k \in K_j$, one of the following holds 
    \begin{enumerate}
        \item $x \geq \eta_j$ for all $x \in T_j^k$.
        \item $x \leq \eta_j$ for all $x \in T_j^k$.
        \item $T_j^k$ is contained in a non-monotone ray originating at $\eta_j$, i.e. $T_j^k \subset \{ x \in \R^I : (x' - \eta_j) = y \kappa, \ y \in \R_+ \}$ for some $\kappa \in \R^I$ such that $\kappa \not \geq 0$ and $\kappa \not \leq 0$.
    \end{enumerate}

    \noindent\textit{Proof of Claim 4.} If either $x'$ or $x''$ are equal to $\eta_j$ then we are done, so assume that this is not the case. Consider the condition from Claim 3. By Farkas' lemma, the following are equivalent
    \begin{enumerate}
        \item there does not exist $c_j \in \R^I_+$ such that $c_j \cdot x' < c_j \cdot \eta_j$ and $c_j \cdot x'' > c_j \cdot \eta_j$. 
        \item There exist $y',y'' \in \R_+$, with at least one strictly positive, such that $y' (x' - \eta_j) \geq y'' (x'' - \eta_j)$. 
    \end{enumerate}
    So Claim 3 implies that the second condition above holds. Reversing the roles of $x'$ and $x''$ we also have that there exist $\hat{y}',\hat{y}'' \in \R_+$, with at least one strictly positive, such that $\hat{y}' (x' - \eta_j) \leq \hat{y}'' (x'' - \eta_j)$.

    Suppose $y' > 0$ (a symmetric argument applies if $y'' > 0$). Then $x' - \eta_j \geq \frac{y''}{y'}(x'' - \eta_j)$. Then we have $\hat{y}''(x'' - \eta_j) \geq \hat{y}'\frac{y''}{y'}(x'' - \eta_j)$. This implies that one of the following holds
    \begin{enumerate}
        \item $x'' - \eta_j \geq 0$ and $\hat{y}'' \geq \hat{y}'\frac{y''}{y'}$
        \item $x'' - \eta_j \leq 0$ and $\hat{y}'' \leq \hat{y}'\frac{y''}{y'}$.
        \item $x'' - \eta_j \not\geq 0$, $x'' - \eta_j \not \leq 0$ and $\hat{y}'' = \hat{y}'\frac{y''}{y'}$. 
    \end{enumerate}
   If we are in the first case, with  $x'' - \eta_j \geq 0$ and $\hat{y}'' \geq \hat{y}'\frac{y''}{y'}$, then $x' - \eta_j \geq 0$ as well. Because $x'$ and $x''$ are arbitrary selections from $T_j^k$, we have that $x \geq 0$ for all $(x - \eta_j) \in T_j^k$. Similarly in the second case, with $x'' - \eta_j \leq 0$ and $\hat{y}'' \leq \hat{y}'\frac{y''}{y'}$, we have that $(x - \eta_j) \leq 0$ for all $x \in T_j^k$. 

   Consider the third case, where $x'' - \eta_j \not\geq 0$, $x'' - \eta_j \not \leq 0$ and $\hat{y}'' = \hat{y}'\frac{y''}{y'}$. Note that if this holds then $y'', \hat{y}',\hat{y}''$ must be strictly positive (and we are still assuming $y' > 0$), so $\frac{\hat{y}''}{\hat{y}'} = \frac{y''}{y'}$. Then we have 
   \[
    \frac{y''}{y'}(x'' - \eta_j) \geq (x' - \eta_j) \quad \text{and} \quad  \frac{y''}{y'}(x'' - \eta_j) \leq (x' - \eta_j)
   \]
   so $ \frac{y''}{y'}(x'' - \eta_j) = (x' - \eta_j)$. In other words, $(x'' - \eta)$ and $(x' - \eta)$ lie in the same ray. Then $T_j^k$ is contained in a ray originating at $\eta_j$.\hfill\qedsymbol

   We call a set $T_j^k$ \textit{remedial} if $x \leq \eta_j$ for all $x \in T_j^k$, and \textit{redundant} if $x \geq \eta_j$ for all $x \in T_j^k$.

   \noindent\textit{Claim 5.} If $T_j^{k'}$ and $T_j^{k''}$ are not remedial then they are polarized.

   \noindent\textit{Proof of Claim 5.} If one of $T_j^{k'}$ and $T_j^{k''}$ is redundant then they are trivially polarized, so suppose that neither is redundant. Then by Claim 4, both $T_j^{k'}$ and $T_j^{k''}$ are contained in rays from $\eta_j$. Thus it suffices to show that there exists $x' \in T_j^{k'}$, $x'' \in T_j^{k''}$, and $y',y'' \geq 0$, with at least one strictly positive, such that $y'x' + y''x'' \in \R^I_+$. 

   Let $x' \in T_j^{k'}$, $x'' \in T_j^{k''}$ be any selections such that $x' \neq \eta_j$, $x'' \neq \eta_j$. Observe that there cannot exist $c_j \geq 0$ such that $c_j \cdot x' \leq c_j \eta_j$ and $c_j \cdot x'' < c_j \eta_j$. If this were the case then, because $x' - \eta_j \not \geq 0$, we could find $c_j'$ such that $c_j' x' < c_j' \eta_j$ and $c_j' x'' < c_j' \cdot \eta_j$, violating Claim 2. By Farkas' lemma, the following are equivalent
   \begin{enumerate}[i.]
       \item There does not exist $c_j \geq 0$ such that $c_j \cdot x' \leq c_j \eta_j$ and $c_j \cdot x'' < c_j \eta_j$.
       \item There exist $y' \geq 0, y'' > 0$, such that $y'(x' - \eta_j) + y'' (x'' - \eta_j) \geq 0$.
   \end{enumerate}
   The second condition is precisely what we wanted to show. \hfill\qedsymbol

   \noindent\textit{Claim 6.} If $T_j^k$ is remedial and $T_j^k \neq \{\eta_j\}$ then $T_j^{k'}$ is remedial or redundant for all $k' \in K_j$.

   \noindent\textit{Proof of Claim 6.} If $T_j^{k'}$ is not remedial or redundant then it is contained in a non-monotone ray from $\eta_j$, by Claim 4. Then there exists $c_j \geq 0$, $c_j \neq 0$, and $x' \in T_j^{k'}$ such that $c_j \cdot x' < c_j \cdot \eta_j$. Then for any $x \in T_j^k$, $x \neq \eta_j$ we can also choose $c_j$ so that $c_j \cdot x < c_j \cdot \eta_j$. But this violates Claim 2. \hfill\qedsymbol

   Say that agent $j$'s mechanism is \textit{remedial} if $T_j^k$ is remedial for all $k \in K_j$ such that $T_j^k \neq \{\eta_j\}$. 

   We have concluded that for any $j \in \mathcal{J}$, the collection of trade sets for $j$ takes one of the following two forms. 
   \begin{enumerate}
       \item $T_j^k$ is remedial or redundant for all $k \in T_j^k$. 
       \item $(T_j^k)_{k \in K_j}$ can be partitioned into a sets contained in a set of polarized rays from $\eta_j$, and a set of redundant sets. 
   \end{enumerate}
   Moreover, if preferences are strictly monotone then no redundant set can ever be chosen (doing so is obviously dominated by choosing $\{\eta_j\}$), and so we obtain an equivalent mechanism by removing the redundant sets. Then the mechanism dominated by a polarized ray mechanism. 

    It remains to show that any polarized ray mechanism is OSP. Any $j$ with a remedial mechanism has a trivial dominant strategy. Suppose $j$'s mechanism consists of polarized rays. Then (as shown using Farkas' lemma in the proof of Claim 5) if $c_j \cdot x' \leq c_j \cdot \eta_j$ for some $x' \in T_j^{k'}$ then $c_j \cdot x'' \geq c_j \cdot \eta_j$ for all $x'' \in T_j^{k''}$, $k'' \neq k'$. Thus choosing $k'$ is obviously dominant. 
\end{proof}

\Cref{thm:osp_trade} significantly reduces the dimensionality of the space of mechanisms over which the designer needs to search. In particular, the characterization implies that $|K_j| \leq I+1$ for all $j$.

\begin{corollary}\label{cor:polarizing_dimension}
Any collection of polarized rays in $\R^I$ has at most $I$ members. Thus $|K_j| \leq I + 1$ in any non-redundant OSP mechanism.   
\end{corollary}
\begin{proof}
    Let $(\kappa_j^k)_{k \in K_j}$ be the vectors defining the polarized rays for $j$. By definition $\kappa_j^{k'} \not \geq 0$ for any $k' \in K_j$, so it must have at least one strictly negative coordinate. Moreover, if $\kappa_j^{k'}(i) < 0$ then $\kappa_j^{k''}(i) \geq 0$, because otherwise $y'  \kappa_j^{k'} + y''  \kappa_j^{k''}  \not\geq 0$ for any $y',y''\geq 0$ with at least one strictly positive. Thus there can be at most $I$ such vectors. 
\end{proof}

Another implication of \Cref{thm:osp_trade} is that the selection rule is irrelevant for determining whether trading mechanism $(\mathcal{T},s_{\mathcal{J}})$ is OSP, aside from determining the mechanism's essential reduction. 

\begin{corollary}\label{cor:selectionrule}
    Assume preferences are strictly monotone and $(\mathcal{T},s_{\mathcal{J}})$ is OSP. Let $(\hat{\mathcal{T}},s_{\mathcal{J}})$ be the polarized ray mechanism that essentially dominates $(\mathcal{T},s_{\mathcal{J}})$ (which exists by \Cref{thm:osp_trade}). Then replacing the selection rule in $(\hat{\mathcal{T}},s_{\mathcal{J}})$ with any alternative rule yields a mechanism that is also OSP. 
\end{corollary}

In light of \Cref{cor:selectionrule}, it is thus without loss of optimality for the designer to use an ex-post-optimal allocation rule of the form 
\[
P(k_{\mathcal{J}}) \in \argmin_{\mu \in \Delta(\mathcal{Z})} \sum_{j \in \mathcal{J}} \pi_j \cdot \mu_j \quad s.t. \quad \mu_j \in T_j^{k_j} \ \ \forall \ j \in \mathcal{J}.
\]

We call a polarized ray mechanism with this selection rule an \textit{ex-post optimal polarized ray} (EOPR) mechanism. Together, the previous observations greatly simplify the problem of choosing a trading mechanism.

\subsection{Robustness}
    Consider any EOPR mechanism with $\eta_j = \sigma_j$. Because, $\eta_j \in T_j^k$ for all $j \in \mathcal{J}$ and $k \in K_j$, the designer always has the option to select the status-quo allocation, regardless of the agents' actions. Thus the mechanism ensures that the expected social cost is no larger than under the status quo. In particular, the mechanism ensures gains in expectation regardless of the distribution over agents' preferences, and is thus robust to misspecification in this dimension. Indeed, this conclusion does not even rely on the agents' preferences being linear.\footnote{If preferences are non-linear then the mechanism may no longer be OSP. However, the designer is still guaranteed a (weak) improvement in any equilibrium of the induced game.} Any EOPR mechanism is thus a reasonable candidate for use in practice. Nonetheless, if the designer has some knowledge of the type distribution they may want to use it to try to select among EOPR mechanisms. 

\section{Towards optimal trading mechanisms}\label{sec:lm}

By \Cref{thm:osp_trade} and \Cref{cor:selectionrule}, the designer can without loss of optimality restrict attention to EOPR mechanisms, i.e. polarized ray mechanisms with the selection rule 
\[
P(k_{\mathcal{J}}) \in \argmin_{\mu \in \Delta(\mathcal{Z})} \sum_{j \in \mathcal{J}} \pi_j \cdot \mu_j \quad s.t. \quad \mu_j \in T_j^{k_j} \ \ \forall \ j \in \mathcal{J}.
\]
Despite the significant dimension reduction, the designer's objective remains somewhat complicated. For a given mechanism and preference profile, OSP renders the agents' best-response problem trivial. However, to determine the allocation the designer needs to solve an additional linear program, which links together the mechanisms of different agents in a way that is not entirely transparent. 

We can simplify this program somewhat by studying a ``large-market'' relaxation. In this relaxation certain parts of the designer's program become separable across agents. A solution to the relaxed program can then be taken as a starting point in the search for solutions to the original program; we show that the solution to the relaxed program is approximately optimal in the original program for large markets. It is not in general possible to solve the relaxed program analytically, and its computational complexity remains an open question. Nonetheless, the parallelization enabled by the large-market relaxation does reduce the computational burden. Additionally, in a special case with two types of tasks the program admits an analytical solution.

\subsection{The large-market model}

In the designer's true program every task must be assigned to a agent for every realized preference profile. Suppose instead that we require only that every task be assigned in expectation, taken over the preference profiles. That is, for a mechanism $m$ we replace the market-clearing condition from \cref{prog:1} with 
\begin{equation}\label{eq:mc_lm}
\E\left[ \sum_{j \in \mathcal{J}}  m_j(i|\mathbf{c}) \right] =  n^i \quad \forall \ i \in \mathcal{I}.
\end{equation}
One interpretation of this relaxation is that each agent is actually a unit mass population of agents with identical performance, and $F_j$ describes the distribution of preferences in this population. We therefore refer to this relaxation as the large-market market-clearing (LM-MC) constraint. 

Fix an endowment $\eta \in \Delta(\mathcal{Z})$. Let $\{T_j^k\}_{k \in K_j}$ be a collection of trading sets for each $j$, with endowment $\eta_j$, and let $s_j: C_j \mapsto K_j$ be a myopically-optimal strategy for $j$. Let $(P_j:K_j \mapsto \R^I_+)_{j \in \mathcal{J}}$ be a collection of \textit{individual-level selection rules} for the designer, where we require that $P_j(k) \in T_j^k$, but do not impose that together $(P_j)_{j\in \mathcal{J}}$ satisfy ex-post market clearing. Then the indirect mechanism in which each agent chooses from $K_j$ according to strategy $s_j$, and the designer chooses from $T_j^{s_j(c_j)}$ according to $P_j$, satisfies the LM-MC constraint if and only if 
\[
\sum_{j \in \mathcal{J}} E_j[P_j(s_j(c_j))] = (n_1, \dots, n_I).
\]
Conversely, if we replace the individual-level selections rules $(P_j)_{j\in \mathcal{J}}$ with a (legitimate) selection rule $P: \prod_{j\in \mathcal{J}} K_j \mapsto \Delta(Z)$, which satisfies ex-post market clearing, then LM-MC is satisfied a fortiori. Thus the following program is an upper bound on the value of choosing the optimal OSP trading mechanism with endowment $\eta$:
\begin{align}
    \min_{(\{T_j^k\}_{k \in K_j}, P_j)_{j \in \mathcal{J}}}  \ \ &\sum_{i=1}^I\sum_{j=1}^J \pi_j(i) \ \E\left[ m_j(i | \mathbf{c}) \right] \label{prog:2} \\
    s.t. \quad & \text{either $j$'s trading sets are polarized rays, or $j$ receives} \notag \\ 
    &\quad \text{ two trading sets: $\{\eta_j\}$ and $\{x \in \R_+ \ : \ x \leq \eta_j\}$; and} \tag{OSP} \\
    &\sum_{j \in \mathcal{J}} E_j[P_j(s_j(c_j))] = (n_1, \dots, n_I).
    \tag{LM-MC}
\end{align}
We can relax this program further by allowing the designer to randomize over the trading sets and selection rule for each $j$. Under the population interpretation of the large-market program, randomization is equivalent to splitting each ``agent'' into subsets that receive different mechanisms. Therefore to economize on notation we don't explicitly model randomization in program (\cref{prog:2}), as it can be represented by redefining $\mathcal{J}$.\footnote{Technically, this notational trick is only valid if the randomization involves rational weights, but this should not cause confusion.} We refer to program in (\ref{prog:2}) (allowing for randomization) as the \textit{large-market OSP program}.

\subsection{Solving the large-market OSP program}

It is convenient to separate out the parts of the program that depend only on $\E_j[P_j(s_j(c_j))]$. Define the \textit{OSP-feasible set} for $j$ by 
\begin{align*}
\mathcal{F}_j &:= \cvh\Big\{ \E_j[P_j(s_j(c_j))] \in \R^I_+ \ : \  P_j:K_j \rightarrow \R^I_+; \text{ and either $j$'s trading sets are polarized rays,} \notag \\ 
    &\quad \text{ or $j$ receives two trading sets: $\{\eta_j\}$ and $\{x \in \R_+ \ : \ x \leq \eta_j\}$}\Big\}.
\end{align*}
where $\cvh$ denotes the convex hull. In other words, $\mathcal{F}_j$ is the set of expected allocations for $j$ that can be induced using a (random) trading protocol and selection rule satisfying all the conditions of program (\ref{prog:2}) aside from (LM-MC). Then the large-market OSP program has the same value as
\begin{align}
    \min_{(x_j)_{j\in \mathcal{J}}} & \   \sum_{j\in \mathcal{J}} \pi_j \cdot x_j \label{prog:outer} \\
    s.t. \quad & x_j \in \mathcal{F}_j \quad \forall \ j \in \mathcal{J} \notag \\
    & \sum_{j \in \mathcal{J}} x_j(i) =  n^i \quad \forall \ i \in \mathcal{I} \tag{LM-MC}
\end{align}
I refer to \cref{prog:outer} as the \textit{outer program}. It is instructive to study the dual of this program, which can be written as  
\begin{equation}\label{prog:dual}
\sup_{\lambda \in \R^I} \sum_{i \in \mathcal{I}} \lambda^in^i - \sum_{j \in \mathcal{J}} S_j(\lambda - \pi_j )
\end{equation}
where $S_j(w) := \max \{ w \cdot y \ : \ y \in \mathcal{F}_j \}$
is the support function of the convex set $\mathcal{F}_j$. Strong duality holds, meaning that the values of the outer program in \cref{prog:outer} and its dual in \cref{prog:dual} coincide. 

\begin{proposition}\label{prop:duality}
     Strong duality for the outer program holds. Moreover, if $\lambda_*$ is a solution to the program in \cref{prog:dual} then there exist selections 
     \[
        x_j \in \argmax \{ (\lambda_* - \pi_j) \cdot y \ : \ y \in \mathcal{F}_j \}
     \]
    such that $(x_j)_{j\in \mathcal{J}}$ satisfy market clearing, and these constitute a solution to the program in \cref{prog:outer}. 
\end{proposition}
\begin{proof}
    Proof in \Cref{proof:duality}.
\end{proof}

Given the support functions $(S_j)_{j \in \mathcal{J}}$ the dual program in \cref{prog:dual} is a simple $I$-dimensional convex minimization program. The only difficulty lies in characterizing the support functions. Writing the program defining these functions more explicitly, we have 
\begin{align*}
    S_j(w) &= \max_{\{T_j^k\}_{k\in K_j}, P_j, s_j} \ \ w \cdot\E_j[P_j(s_j(c_j))] \\
    s.t. \quad & s_j \text{ is myopically optimal;} \\
    & P_j: K_j \rightarrow \R^I_+ \text{ ; and} \\
    & \text{either $j$'s trading sets are polarized rays, or $j$ receives} \notag \\ 
    &\quad \text{ two trading sets: $\{\eta_j\}$ and $\{x \in \R_+ \ : \ x \leq \eta_j\}$}.
\end{align*}
Ultimately this program needs to be solved computationally. The fact that the support function can be characterized in parallel for each agent is of some assistance, but the problem remains challenging. Still, we can simplify the program somewhat with a few observations. First, note that the option of giving $j$ the remedial menu, i.e. the two trading sets $\{\eta_j\}$ and $\{x \in \R_+ \ : \ x \leq \eta_j\}$ is only really valuable if $w \leq 0$, in which case $S_j(w) = 0$.\footnote{Formally, if $w \not\leq 0$ then $\max \{w \cdot x : x \leq \eta_j \}$ is solved by setting $x^*(i) = \eta_j(i)$ if $w(i) \geq 0$, and $x^*(i) = 0$ otherwise. Then let $\kappa_j = x^* - \eta_j$. The trading protocol with sets $\{\{\eta_j\}, \{y \kappa_j + \eta_j, y \in \R_+$\}\} yields the same outcomes as the remedial menu. These are not quite polarized rays, as $\kappa_j \leq 0$. However, we can make it a polarized ray mechanism adding $\varepsilon > 0$ to $\kappa_j(i)$ if $w^*(i) > 0$. If $j$'s type is continuously distributed then this trading mechanism yields approximately the same outcomes. Alternatively, we can extend the definition of polarized ray mechanisms to allow for $\kappa_j(i) \leq 0$. Or, we could simply solve for the optimal polarized ray mechanism and compare it to the value of $\max \{w \cdot x : x \leq \eta_j \}$.} Thus the only difficult case is $w \not\leq 0$, where we solve 
\begin{align*}
S_j(w) &= \max_{\{T_j^k\}_{k\in K_j}, P_j, s_j} \ \ w \cdot\E_j[P_j(s_j(c_j))] \\
    s.t. \quad & s_j \text{ is myopically optimal;} \\
    & P_j: K_j \rightarrow \R^I_+ \text{ ; and} \\
    & \text{$j$'s trading sets are polarized rays}
\end{align*}
Now consider the constraint that $j$'s trading sets must be polarized rays. We have already observed that there are at most $I$ such rays (\Cref{cor:polarizing_dimension}). Let $\kappa_j^k$ be the vector defining the ray $T_j^k$ (with $\kappa_j^k = 0$ for $T_j^k = \{\eta_j\}$). The the polarized-ray constraint is
\[
\cvh\big\{ \kappa_j^{k'},\kappa_j^{k''}\big\}\cap\R^I_+ \neq \varnothing \quad \forall \ k',k'' \in K_j.
\]
We can also re-write the polarizing ray constraint by applying Farkas' Lemma.
\begin{lemma}\label{lem:polarized_alt}
    The rays $\{T_j^k\}_{j \in K_j}$ defined by vectors $\{\kappa_j^k\}_{k \in K_j}$ are polarized if and only if 
    \[
    \min_{c_j \in R^I_+} \left\{ c_j \cdot \kappa_j^{k'} \ : \ c_j \cdot \kappa_j^{k''} \leq 0  \right\} \geq 0 \quad \forall \ k' \neq k'' \text{ with } T_j^{k''} \neq \{\eta_j\}.
    \]
\end{lemma}
\begin{proof}
    Proof in \Cref{proof:polarized_alt}.
\end{proof}

Despite these simplifications, the extent to which $S_j$ can be efficiently characterized remains an open question, which we leave for future work. 

\begin{remark}
    \cite{baron2025mechanism} study a version of the large-market program with only two types of task, but allowing for any BIC mechanism. However, they show that under a regularity condition on the type distribution, the optimal mechanism is in fact a polarized-ray mechanism (with two rays). In this case, the restriction to OSP trading mechanisms is without loss of optimality (in the large-market program) and the support function can be characterized analytically. 
\end{remark}

\subsection{From the large market back to the small}

As discussed above, any EOPR mechanism improves upon the status quo, both from the designer's perspective and from that of every agent. A solution to the large-market program is a natural candidate. This mechanism is (unsurprisingly) approximately optimal with a large finite number of agents. To formalize this claim, consider a sequence of replica economies, indexed by an integer $N$. In the $N$-replica economy, there are $N$ copies of each agent with identical performance, but with preferences drawn independently from the same distribution. We refer to these agents as ``group $j$''. There are also $N n^i$ copies of task $i$.

Fix a solution to the large-market program, which is an EOPR mechanism that we denote by $m^*$. For each $N$, we implement the mechanism $m^*$. If $m^*$ involves randomization for agent $j$ (as permitted in the large-market program) then we replicate the randomization in the $N$-replica economy by offering each trading protocol in the support of the randomized mechanism to a proportional fraction of the $N$ group-$j$ agents.\footnote{That is, if under $m^*$ agent $j$ receives trading protocol $\{T_j^k\}_{k \in K_j}$ with probability $\varepsilon$, then $\lfloor \varepsilon N \rfloor$ of the group-$j$ agents receive this trading protocol (any residual agents can be allocated to any mechanism in the support of $m^*$, as this will not affect the asymptotic results).} Let $V^*_{LM}(N)$ be the random variable corresponding to the social cost under this mechanism, and let $V^*(N)$ be the corresponding random variable for the optimal mechanism in the $N$-replica economy (they are random because they depend on the realized preference profile). We say that the solution to the large-market program is \textit{asymptotically optimal as the market grows} if $\frac{1}{N}V^*_{LM}(N) \xrightarrow{a.s.} \frac{1}{N} V^*(N)$ as $N \rightarrow \infty$.  

\begin{proposition}\label{prop:asymptotic}
    The optimal mechanism in the large-market program is asymptotically optimal in the original program as the market grows.
\end{proposition}
\begin{proof}
    Proof in \Cref{proof:asymptotic}.
\end{proof}

% \section{Simulations}

\section{Conclusion}\label{sec:conclusion}

This paper studies the problem of assigning heterogeneous tasks (or objects) across agents when agents’ preferences are private information and each agent can insist on a status-quo assignment. The first contribution is to provide a sharp characterization the conditions under which it is optimal for the designer to relinquish control in order to allow assignments to respond to agents’ preferences via choice-based mechanisms. \Cref{thm:choice} shows that, under mild conditions on the status quo and on the richness of preference heterogeneity, choice has value whenever the underlying allocation problem is non-trivial: whenever agent performance is not identical and preference distributions are non-degenerate, there exist mechanisms that strictly improve on the status quo. Moreover, the result establishes that in terms of determining whether choice is valuable, restricting attention to obviously strategy-proof trading mechanisms is without loss relative to the full class of Bayesian incentive compatible mechanisms. 

The second contribution is an explicit and practically useful characterization of the class of and obviously strategy-proof trading mechanisms. Theorem 2 shows that all OSP trading mechanisms are effectively polarized ray mechanisms: agents either receive no effective choice, or choose among a menu of polarized rays emanating from the status quo. This characterization delivers a dimension reduction both in terms of the shape of trading sets and the size of the menu. Moreover, such mechanisms are robust: ex-post optimal polarized ray (EOPR) mechanisms guarantee weak improvements over the status quo regardless of the distribution of preferences. Finally, the paper provides guidance for implementation by developing a large-market relaxation with a (relatively) tractable dual formulation and establishing that the resulting design is asymptotically optimal in large replica economies, thereby offering a principled route for selecting among OSP candidates in practice.

A number of interesting questions remain unanswered. For instance, how does the characterization of OSP mechanisms change when we allow for transfers? The characterization is likely to change more in settings where tasks are ``goods'' (see \Cref{sec:transfers})? Another question is how we can identify EOPR mechanisms that perform well according to the designer's objective. The dual formulation of the large-market program simplifies this problem, as discussed in \Cref{sec:lm}. However, this approach requires characterizing the support function $S_j$, which entails solving a difficult optimization problem in which the domain of choice is menus of polarized rays. Computational approaches to approximately solve this program would be useful in practice.

\newpage
\bibliography{references}

\newpage
\appendix
\section{Omitted proofs}

\subsection{Proof of \texorpdfstring{\Cref{prop:duality}}{}}\label{proof:duality}

\begin{proof}
We first verify that the program in \cref{prog:dual} is indeed the dual of \cref{prog:outer}. Letting $\lambda^i$ be the multipliers on the LM-MC constraints, we have
\begin{align*}
    \min_{(x_j)_{j\in \mathcal{J}}}& \sup_{\lambda \in \R^I} \ \sum_{j\in \mathcal{J}} \pi_j \cdot x_j - \sum_{i\in \mathcal{I}}\lambda^i \left(\sum_{j \in \mathcal{J}} x_j(i) -  n^i\right)
    \quad s.t. \quad x_j \in \mathcal{F}_j \quad \forall \ j \in \mathcal{J} \\
    & \geq  \sup_{\lambda \in \R^I} \min_{(x_j)_{j\in \mathcal{J}}} \ \sum_{j\in \mathcal{J}} \pi_j \cdot x_j - \sum_{i\in \mathcal{I}}\lambda^i \left(\sum_{j \in \mathcal{J}} x_j(i) -  n^i\right)
    \quad s.t. \quad  x_j \in \mathcal{F}_j \quad \forall \ j \in \mathcal{J} \\
    &= \sup_{\lambda \in \R^I} \  \sum_{i \in \mathcal{I}} \lambda^in^i - \sum_{j\in \mathcal{J}} \max_{x_j\in \mathcal{F}_j} \sum_{i \in \mathcal{I}} \left(\lambda^i - \pi_j(i)\right) x_j^i \\
    &= \sup_{\lambda \in \R^I} \ \sum_{i \in \mathcal{I}} \lambda^in^i - \sum_{j \in \mathcal{J}} S_j(\lambda - \pi_j)
\end{align*}
where inequality in the second line follows from weak duality, and the last line from the definition of the support function. 

Because giving the single trading set $\{\eta_j\}$ to agent $j$ is feasible, we have $\eta_j \in \mathcal{F}_j$. If $\eta_j$ is in the interior of $\mathcal{F}_j$ for all $j$ then Slater's condition is satisfied, so we are done. Otherwise, define a perturbed problem by letting
\[
\mathcal{F}_j^{\varepsilon} = \{x \in \R^I_+ \ : \exists \ y \in \mathcal{F}_j \text{ with } |x - y| \leq \varepsilon\}.
\]
for $\varepsilon \geq 0$, and let $S_j^{\varepsilon}$ be the support function of this set. Then $\eta_j$ is in the interior of $\mathcal{F}_j^{\varepsilon}$ for all $j$, so Slater's condition is satisfied for any $\varepsilon > 0$ if we replace $\mathcal{F}_j$ with $\mathcal{F}_j^{\varepsilon}$. Let $P_{\varepsilon}$ and $D_{\varepsilon}$ be the values of the primal and dual programs at $\varepsilon$. Notice that $\mathcal{F}_j^{\varepsilon}$ is compact and convex valued and $\varepsilon \mapsto \mathcal{F}_j^{\varepsilon}$ is upper- and lower-hemicontinuous on $\R_+$. Thus $\varepsilon \mapsto P_{\varepsilon}$ is continuous by Berge's maximum theorem, and converges to $P_0$ as $\varepsilon \rightarrow 0$.

Now note that $\varepsilon \mapsto S_j^{\varepsilon}(w)$ is increasing, so $\varepsilon \mapsto D_{\varepsilon}$ is decreasing. Moreover, weak duality (for a minimization problem) implies that $P_0 \geq D_0$. Then $P_{\varepsilon} = D_{\varepsilon} \leq D_{0} \leq P_0$ for all $\varepsilon > 0$, and $P_\varepsilon \rightarrow P_0$ as $\varepsilon \rightarrow 0$, which implies $D_0 = P_0$.  
\end{proof}

\subsection{Proof of \texorpdfstring{\Cref{lem:polarized_alt}}{}}\label{proof:polarized_alt}

\begin{proof}
    The rays defined by $\kappa_j^{k'}, \kappa_j^{k''}$ are polarized (by definition) if and only if there exists scalars $y',y'' \geq 0$, with at least one strictly positive, such that $y'\kappa_j^{k'} + y'' \kappa_j^{k''} \geq 0$. Assume without loss of generality that $y' > 0$. By Farkas' Lemma, the following are equivalent 
    \begin{enumerate}[i.]
       \item There does not exist $c_j \in \R^I_+$ such that $c_j \cdot \kappa_j^{k''} \leq 0$ and $c_j \cdot \kappa_j^{k'} < 0$.
       \item There exist $y'' \geq 0, y' > 0$, such that $y'\kappa_j^{k'} + y'' \kappa_j^{k''} \geq 0$.
   \end{enumerate}
   The first condition can be re-stated as: $c_j \cdot \kappa_j^{k'} \geq 0 \ \ \forall \ c_j \in \R^I_+$ such that $c_j \cdot \kappa_j^{k''} \leq 0$, which is what we wanted to show.
\end{proof}

\subsection{Proof of \texorpdfstring{\Cref{prop:asymptotic}}{}}\label{proof:asymptotic}

\begin{proof}
To simplify notation, assume that the large-market mechanism does not involve randomization. It is straightforward to extend the argument to the case with randomization by sub-dividing the groups.

As the mechanism is OSP, the strategy of each agent is independent of $N$. Let $d_j^k(N)$ be the random variable corresponding to the set of group-$j$ agents who choose trading set $T_j^k$ in the $N$-replica economy. The mechanism allows the designer can choose different allocations for agents in $d_j^k(N)$. However, towards establishing asymptotic convergence we can assume that (perhaps sub-optimally) the designer does not do so. Define the random variable 
\begin{align*}
    \tilde{V}^*_{LM}(N) = \min_{(\mu_{j}^k)} \sum_{j \in \mathcal{J}} \sum_{k \in K_j} |d_j^k(N)| \pi_j^k \cdot \mu_j^k \quad s.t. \quad &\mu_j^k \in T_j^k \ \forall \ j \in \mathcal{J}, \ k \in K_j, \\ 
     & \frac{1}{N} \sum_{j \in \mathcal{J}} \sum_{k \in K_j} \mu_j^k |d_j^k(N)| = (n_1,\dots,n_I)
\end{align*}
Let $V^*_{LM}(\infty)$ be the value obtained in the large-market program, which is a lower bound on $\frac{1}{N}V^*(N)$ for any $N$. Then to prove the result it suffices to show that $\frac{1}{N} \Tilde{V}^*_{LM}(N) \rightarrow V^*_{LM}(\infty)$ almost surely as $N \rightarrow \infty$. Define the random variable $\rho_j^k(N) := \frac{|d_j^k(N)|}{N}$ (note that $\sum_{k \in K_j} \rho_j^k = 1$ for all $j$). The we can write
\begin{align}
    \frac{1}{N}\tilde{V}^*_{LM}(N) &= \min_{(\mu_{j}^k)} \sum_{j \in \mathcal{J}} \sum_{k \in K_j} \rho_j^k(N) \pi_j^k \cdot \mu_j^k \label{prog:asymptotic} \\
    s.t. \quad &\mu_j^k \in T_j^k \ \forall \ j \in \mathcal{J}, \ k \in K_j, \label{eq:replica_PR} \\ 
     & \sum_{j \in \mathcal{J}} \sum_{k \in K_j} \mu_j^k \rho_j^k(N) = (n_1,\dots,n_I) \label{eq:replica_MC}
\end{align}
Abusing notation, let $\Tilde{V}^*_{LM}(\rho)$ be the value of this program as a function of $\rho = (\rho_j^k)_{j \in \mathcal{J}, k \in K_j}$. If $\rho_j^k = \E_j[\mathbbm{1}\{ s_j(c_j) = k \}]$ then $\Tilde{V}^*_{LM}(\rho) = V^*_{LM}(\infty)$. By the strong law of large numbers $\rho_j^k(N) \xrightarrow{a.s.} \E_j[\mathbbm{1}\{ s_j(c_j) = k \}]$ as $N \rightarrow \infty$ for all $j,k$. Thus we just need to show that $\Tilde{V}^*_{LM}(\rho)$ is upper-semicontinuous at $\E_j[\mathbbm{1}\{ s_j(c_j) = k \}]$. 

Let $\kappa_j^k$ be a vector defining the ray $T_j^k$, chosen so that the non-negativity constraint binds at $\eta_j + \kappa_j^k$. Then any $\mu_j^k$ satisfying constraint (\ref{eq:replica_PR}) can be written as $\mu_j^k = \eta_j + \alpha^j_k \kappa_j^k$ for some scaling factors $\alpha_j^k \in [0,1]$. Let $(\hat{\alpha}_j^k)$ be the scaling factors that define the solution to the large-market program. Without loss of generality, we can assume that $\Hat{\alpha}_j^k > 0$, since otherwise $T_j^k$ can just be replaced by ${\eta_j}$. Note also that the allocation defined by $(\delta \hat{\alpha}_j^k)$ also satisfies market clearing for any $\delta \in [0,1]$, because the resulting mechanism is just a convex combination of the large-market solution and the endowment $\eta$, and thus $(\delta \hat{\alpha}_j^k)$ is also a valid mechanism in the large-market program. Then for any $\varepsilon \geq 0$ we can choose $\delta < 1$ so that the value of $(\delta \hat{\alpha}_j^k)$ in the large-market program, denote it by $V^*_{\delta}$, is no more than $V^*_{LM}(\infty) + \varepsilon$. Then to show that $\Tilde{V}^*_{LM}(\rho)$ is upper-semicontinuous at $\E_j[\mathbbm{1}\{ s_j(c_j) = k \}]$ it suffices to show that for any $\varepsilon > 0$ and $\delta \in (0,1)$ there exists $\sigma > 0$ such that $|\E_j[\mathbbm{1}\{ s_j(c_j) = k \}] - \rho_j^k | < \sigma$ implies $\Tilde{V}_{LM}^*(\rho) \leq V^*_{\delta} + \varepsilon/2$. 

As the objective is continuous in $\rho$, we just need to show that by taking a $\rho_j^k$ sufficiently close to $\E_j[\mathbbm{1}\{ s_j(c_j) = k \}]$ we can ensure that there exists $(\alpha_j^k)$ such that the resulting mechanism satisfies the market-clearing constraint (\ref{eq:replica_MC}), and such that $(\alpha_j^k)$ is arbitrarily close to $(\delta \hat{\alpha}_j^k)$. As a function of $(\alpha_j^k)$ the market clearing constraint is 
\[
    \sum_{j \in \mathcal{J}} \eta_j + \sum_{k \in K_j} \alpha_j^k \kappa_j^k \rho_j^k = (n_1,\dots,n_I) 
\]
As discussed above, $(\delta \Hat{\alpha}_j^k)$ satisfies market clearing in the large-market program, meaning the previous condition is satisfied if $\alpha_j^k = \delta \Hat{\alpha}_j^k $ and $\rho_j^k = \E_j[\mathbbm{1}\{ s_j(c_j) = k \}]$. Then it is also satisfied if $\alpha_j^k = \frac{\delta}{1+\gamma_j^k} \Hat{\alpha}_j^k $ and $\rho_j^k = (1 + \gamma_j^k)\E_j[\mathbbm{1}\{ s_j(c_j) = k \}]$. Moreover, as $\delta < 1$ we have $\frac{\delta}{1+\gamma_j^k} \Hat{\alpha}_j^k \leq 1$ as long as $|\gamma_j^k|$ is sufficiently small, so that $(\alpha_j^k)$ are valid scaling factors (i.e. the resulting mechanism does not violate the non-negativity constraint). 
\end{proof}

\end{document}